\documentclass[journal, letterpaper, twocolumn]{IEEEtran}
\usepackage{color}
\usepackage[T1]{fontenc}
\usepackage[latin9]{inputenc}
\usepackage{amsmath}
\usepackage{amssymb}
\usepackage{graphicx}
\usepackage{algorithmic}
\usepackage{algorithm}
\usepackage{float}
\usepackage{placeins}
\usepackage{epstopdf}
\usepackage{comment}
\usepackage{citesort}

\usepackage{amsthm}

\newtheorem{theorem}{Theorem}
\newtheorem{lemma}{Lemma}
\newtheorem{proposition}[theorem]{Proposition}

\theoremstyle{definition}

\theoremstyle{remark}

\begin{document}

\title{Optimizing Joint Data and Power Transfer in Energy Harvesting Multiuser Wireless Networks}
\author{Bassem~Khalfi,~\IEEEmembership{Student~Member,~IEEE,}
        Bechir~Hamdaoui,~\IEEEmembership{Senior~Member,~IEEE,}
        Mahdi~Ben~Ghorbel,~\IEEEmembership{Member,~IEEE,}
        Mohsen~Guizani,~\IEEEmembership{Fellow,~IEEE,}
        Xi~Zhang,~\IEEEmembership{Fellow,~IEEE,}
        and~Nizar~Zorba,~\IEEEmembership{Senior~Member,~IEEE,}
\thanks{Copyright (c) 2015 IEEE. Personal use of this material is permitted. However, permission to use this material for any other purposes must be obtained from the IEEE by sending a request to pubs-permissions@ieee.org}
\thanks{Bassem Khalfi and Bechir Hamdaoui are with Oregon State University, Corvallis, OR, USA (e-mail: khalfib@oregonstate.edu).}
\thanks{Mahdi Ben Ghorbel, Mohsen Guizani and Nizar Zorba are with Qatar University, Doha, Qatar.}
\thanks{Xi Zhang is with Texas A\&M University, College Station, TX, USA.}
\thanks{This work was supported in part by the US National Science Foundation (NSF) under NSF award CNS-1162296. 
}}

\maketitle
\begin{abstract}
Energy harvesting emerges as a potential solution for prolonging the lifetime of the energy-constrained mobile wireless devices.
In this paper, we focus on Radio Frequency (RF) energy harvesting for multiuser multicarrier mobile wireless networks.
Specifically, we propose joint data and energy transfer optimization frameworks for powering mobile wireless devices through RF energy harvesting. We introduce a power utility that captures the power consumption cost at the base station (BS) and the used power from the users' batteries, and determine optimal power resource allocations that meet data rate requirements of downlink and uplink communications. Two types of harvesting capabilities are considered at each user: harvesting only from dedicated RF signals and hybrid harvesting from both dedicated and ambient RF signals.
The developed frameworks increase the end users' battery lifetime at the cost of a slight increase in the BS power consumption.
Several evaluation studies are conducted in order to validate our proposed frameworks.
\end{abstract}


\section{Introduction}
\label{sec:intro}
Minimizing energy consumption and prolonging network lifetime have become primal design goals of next-generation wireless networks, merely due to limited power resources of wireless devices.
Wireless Energy Transfer (WET) technology emerges as a key solution for addressing such issues, and has recently attracted lots of research attention~\cite{varshney2008transporting,shi2011renewable,xie2012renewable,xie2012making,xie2013bundling,Bin2010,Khoshabi2016}.
%
WET technology has even greater impact when considering battery-powered wireless devices whose batteries cannot (or are difficult to) be replaced, as in the case of remote sensor nodes. In addition to carrying the energy, a new paradigm, called Simultaneous Wireless Information and Power Transfer (SWIPT), has recently emerged to allow for distant powering of devices during ongoing data communications~\cite{varshney2008transporting,Kaibin2015,ju2014throughput,ng2013wireless}.

There are three proposed SWIPT design schemes: decoupled SWIPT, closed-loop SWIPT, and integrated SWIPT~\cite{Kaibin2015}. In decoupled SWIPT, the information and the power are sent from two separate sources that could be placed at different locations: a base station represents the information gateway and a power beacon represents the energy gateway. The closed-loop SWIPT scheme powers the device in the downlink and sends the data in the uplink. This scenario could be the case of data offloading in wireless sensor networks, where the main concern is how to offload the data from the sensors~\cite{ju2014throughput}.
In the third design scheme, both the information and power are sent by the same source over the same signals~\cite{ng2013wireless}. However, the challenge lies on how to separate the data and power streams.

Broadly speaking, there are two energy harvesting techniques in single-input single-output systems: time switching and power splitting~\cite{Xiao2015,Liang2013,Gu2015}. Time switching consists of splitting the time window into two portions, where during the first portion, the receiver converts the received RF signals into power while the second portion is dedicated to decoding the RF signals. Although simple, this technique requires a perfect synchronization; otherwise, it induces some information loss~\cite{Krikidis2014}.
Power splitting, on the other hand, consists of splitting the received signal into two streams. The first serves for extracting power and the second for decoding the received information. The splitting ratio balances between the amounts of harvested power and the achieved data rate.

While lots of works focused either on optimizing the power allocation at the base station (BS) or on exploring the users' achieved data rates, the excessive use of power at the BS as well the available battery levels at the different users were not accounted for. In this work, we develop SWIPT techniques that account for power costs at the BS and battery energy available at the different users while harvesting RF energy from not only intended signals but also all nearby ambient RF signals (i.e., interference) intended for other users.
%

\subsection{Related Works}
Varshney et al.~\cite{varshney2008transporting} are among the first researchers that highlighted the potential of transferring energy through RF signals. Since then, WET and SWIPT through RF signals have attracted numerous works~\cite{varshney2008transporting,shi2011renewable,xie2012renewable,xie2012making,xie2013bundling,Xiao2015,Liang2013,Gu2015,Krikidis2014,Zhiguo2014,Zhaoxi2015,Xiao2015,Gregori2014,Zeng2015,Rubio2015,Hoang2015,Zhou2013,Chen2016}.
The authors in~\cite{shi2011renewable,xie2012renewable,xie2012making} proposed an interesting idea for transferring energy wirelessly to sensor network nodes. The idea is basically to have a designated wireless charging vehicle (WCV) that periodically travels inside the network to wirelessly charge sensors' batteries. They formulated an optimization problem whose objective is to maximize the ratio of the WCV's vacation time over the cycle time, and proved that the optimal traveling path for the WCV is the shortest Hamiltonian cycle. This idea has been further applied to networks with mobile base stations~\cite{xie2013bundling}. The authors in~\cite{xie2013bundling} studied the problem of whether and how the mobile BS can be co-located on the WCV to also serve as a charging vehicle. The authors formulated the co-location problem as an optimization problem while accounting for energy charging, WCV's stopping behavior, and data flow routing. Then, they proposed a formulation that depends only on location to serve as a simpler alternative for solving the same general problem. However, WCV can only charge a limited number of sensors at a given time, making the approach unscalable especially when considering large areas. 

There have also been some research efforts studying the performance of RF energy harvesting~\cite{Zhiguo2014,Zhou2013,Zhaoxi2015,Bassem2016}. For instance, the authors in~\cite{Zhiguo2014} investigated energy harvesting in cooperative networks, where a number of source-destination pairs are communicating with each other through an energy harvesting relay. This work proposed power splitting strategies that the relay can use to distribute the harvested energy among multiple users.
In~\cite{Zhou2013,Zhaoxi2015,Bassem2016},
performance tradeoffs between the power-splitting and the time-switching methods, when used for jointly transferring energy and data in various point-to-point systems, have been studied. For example, authors in~\cite{Zhaoxi2015} derived suboptimal power splitting ratios for point-to-point multi-channel systems. In~\cite{Bassem2016}, we investigated the minimization of the system total power while accounting for the received interference at each user.

Energy harvesting has also been studied in the context of multiuser access~\cite{Gregori2014,ju2014throughput,zhou2016wireless}, MIMO systems~\cite{Zeng2015,Rubio2015,Lam2016}, and cognitive radio networks~\cite{Hoang2015,Zhai2016,Khoshabi2016}.
In~\cite{ju2014throughput}, the authors tackled closed-loop SWIPT in a multiuser system, where the optimal time allocation for each user maximizing the sum rate is derived. OFDM access has been considered as well in~\cite{ng2013wireless} with the objective of maximizing the energy efficiency.

Unlike previous works, we consider optimizing the power consumption in the downlink and uplink of a multiuser multi-carrier system with simultaneous information and power transfers. The power utility includes the power cost at the BS required to communicate with the different users, as well as the amount of battery energy available at the users.
Our approach integrates SWIPT with power splitting to increase spectrum efficiency, and allows each user to harvest not only from its dedicated signal, but also from ambient RF signals resulting from the communication between the BS and the other users.

\subsection{Summary of the Contributions}

The main contributions of this paper are:

\begin{itemize}
\item We develop joint data and energy transfer optimization frameworks for wirelessly powering mobile devices via RF energy harvesting. Unlike previous works, we propose a weighted power cost that captures the consumed power at the BS and the battery power available at the users.

\item We analytically derive closed-form expressions of the optimal power allocations required for meeting the data rate requirements of the downlink and uplink communications between the BS and its mobile users.

\item We study two system setups: $(i)$ Users can only harvest energy from their intended/dedicated RF signals; and $(ii)$ Users can harvest energy from any ambient RF signals intended for any user.

\end{itemize}

\subsection{Roadmap}
The rest of this paper is organized as follows. In Section~\ref{sec:sysmod}, we present our system model. We formulate and solve the studied energy harvesting optimization problem in Section~\ref{sec:probfor1111} for the case of dedicated RF signal-based energy harvesting, and in Section~\ref{sec:probfor1} for the case of hybrid dedicated and ambient RF signal-based energy harvesting. Our results are presented in Section~\ref{sec:simulationResults}, and our conclusions are provided in Section~\ref{sec:conclusion}.

\section{System Model}\label{sec:sysmod}
We consider a point-to-multipoint, half-duplex, OFDM network with a BS at the center of a cell and $K$ mobile users, as illustrated in Fig.~\ref{fig:SysMod}.
The BS transmits over $L$ orthogonal subcarriers with only $N$ subcarriers are used to communicate with each user. We assume that the number $N$ is the same for all users, i.e., $L=K\times N$. Without loss of generality, we assume that the first $N$ subcarriers are used to communicate with the first user, the second $N$ subcarriers are used to communicate with the second user and so on. In the uplink, each user adopts SC-FDMA and communicates with the BS over $N$ subcarriers. The downlink and uplink channels between the BS and the $k^{th}$ user over the $i^{th}$ subcarrier are $h_{BS,k}^i$ and $h_{k,BS}^i$, respectively. Note that we defined the uplink and the downlink channels to be different so that our frameworks can fit both TDD and FDD modes. It is also assumed that the BS has perfect knowledge of the different channel gains. We consider that the BS uses the integrated SWIPT to power and communicate with users, and each user relies on the power splitting technique to separate the power and the information streams. We illustrate the high-level receiver's architecture of each device in Fig.~\ref{fig:ReceiverDesign}. The received RF signal affected by the receiver's noise is split into two portions: a first portion is directed to the energy harvesting unit while the second portion is fed to the data processing unit.
This paper's focus is on power allocation in multicarrier energy harvesting wireless systems. Subcarrier scheduling is beyond the scope of this work (see~\cite{Aggarwal2011,Yaacoub2012} if interested).
\begin{figure}
\centering{
\includegraphics[width=.8\columnwidth]{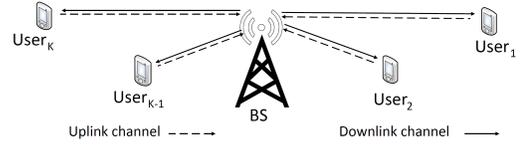}
\caption{System model: a base station and $K$ mobile users.}
\label{fig:SysMod}}
\end{figure}

The communication process adopts the model of~\cite{Zhaoxi2015}. During the first half of time slot $t$, the BS communicates with all the users over the non-interfering subcarriers using a total power $\sum_{k=1}^K\sum_{i=1}^{N}P_{BS,k}^i$, where $P_{BS,k}^i$ is the power used in the downlink to communicate with user $k$ over the $i^{th}$ subcarrier. In the second half of time slot $t$, each user relies on the power splitting technique~\cite{Liang2013} to harvest part of the received RF signal power and uses it, in addition to its remaining battery power ${P}_{k}^{bat}(t)$, to communicate back with the BS. The battery power's level changes over time as
${P}_{k}^{bat}(t+1)={P}_{k}^{bat}(t)+Q_{k}(t)-P_{k}^{proc}(t)$,
where $P^{bat}(t)$ is the available power at the beginning of time slot $t$, $Q_{k}(t)$ is the harvested power, and $P_{k}^{proc}(t)$ is the power used for information processing. This model covers the case where the users are battery-free, which corresponds to ${P}_{k}^{bat}(t)=0$. In the rest of the paper, we drop the time index of the time slot as we are concerned with the optimal power at each time slot.
The signals received by user $k$, $y_{BS,k}^i$ and by the BS, $y_{k,BS}^i$  can be expressed as
\begin{subequations}\label{eqn:signals}
\begin{align}
&y_{BS,k}^i=x_{k}^i\sqrt{P_{BS,k}^i}h_{BS,k}^i+n_{BS,k}^{i}, \label{eqn:signals1a}\\
&y_{k,BS}^i=z_{k}^i\sqrt{P_{k,BS}^i}h_{k,BS}^i+n_{k,BS}^{i}, \label{eqn:signals1b}
\end{align}
\end{subequations}
for $i\in[1,...,N]$ where $x_{k}^i$ and $z_{k}^i$ are the unit-power symbols transmitted by the BS and the $k^{th}$ user, respectively. $P_{BS,k}^i$ and $P_{k,BS}^i$ are the transmission powers at the $i^{th}$ subcarrier in the downlink and the uplink, respectively. $n_{BS,k}^{i}$ and $n_{k,BS}^{i}$ are Additive White Gaussian Noises (AWGN) with zero mean and variance $\sigma_{BS,k}^{i}$ and $\sigma_{k,BS}^{i}$. We consider $\sigma_{k,BS}^{i}=\sigma_{BS,k}^{i}=N_0B$ where $N_0$ is the noise power density.

When using the power splitting approach, the amount of harvested energy at the mobile user $k$ is then expressed as
$Q_{k}=\eta \rho_k (\sum_{i=1}^NP_{BS,k}^i|h_{BS,k}^i|^2+\sigma_{BS,k}^{i})$,
%
where $\eta$, $0<\eta<1$, is the energy harvesting efficiency that is characteristic of the RF circuitry.
$\rho_k$ is the power splitting ratio that balances between the amount of the RF signal used for harvesting energy and the RF signal used to decode the sent signal.
The user considers the second stream for information decoding. A noise term is added at the decoding unit which leads to an achieved rate by the BS of
$R_{BS,k}=\sum_{i=1}^NB\log_2(1+\frac{(1-\rho_k)P_{BS,k}^i{|h_{BS,k}^i|}^2}{\sigma_{BS,k}^i})$,
%
where $B$ is the bandwidth of each sub-band. We assume that all the sub-bands are equal. To simplify the analysis, we assumed that $\sigma_{BS,k}^i\approx(1-\rho_k)\sigma_{BS,k}^i+\sigma_2^2$ where $\sigma_2^2$ is the power of the noise term introduced at the decoding unit.
\begin{figure}[t]
\centering{
\includegraphics[width=.6\columnwidth]{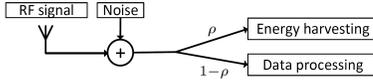}
\caption{Power splitting based receiver structure.}
\label{fig:ReceiverDesign}}
\end{figure}
In the uplink, the mobile user $k$ uses the amount of the harvested power $Q_{k}$, along with its remaining battery power ${P}_{k}^{bat}(t)$, to communicate with the BS. Using Equation~\eqref{eqn:signals1b}, the achieved rate in the uplink can be expressed as
$R_{k,BS}=\sum_{i=1}^NB\log_2(1+\frac{P_{k,BS}^i{|h_{k,BS}^i|}^2}{\sigma_{k,BS}^i})$.


Optimizing the transmit power at the BS and at the users while satisfying some data rate constraints over a long-term interval has its advantages and disadvantages. While power saving can be achieved by taking advantage of the batteries' dynamic, acquiring the channel CSI ahead of time can be very challenging. This is due to the inherent time-varying nature of the wireless channel. On the other hand, optimizing the power instantaneously allows to achieve optimal power allocation by exploiting all the available information (channels' gains, power cost, battery levels, etc). Therefore, we focus on determining the optimal power levels that should be allocated by the BS and each device so that both the BS's and users' data rate requirements are met.
We consider two system setups: i) each user can only harvest energy from its dedicated RF subcarrier signals over which it is receiving its data from the BS, and ii) we extend the harvesting capability to the case where each user can also take advantage of the downlink channels of the other users and harvest energy from any ambient RF signals communicated between the BS and other users.
In the next section, we consider the first system setup and in the following section, we elaborate the second setup.

\section{Dedicated RF Signal Based Energy Harvesting}
\label{sec:probfor1111}

The focus of this section is to optimize a power utility for the whole system. The utility function balances between two entities: the cost of the power that the BS will use to communicate with all the users, and the amount of power used by each user from its battery.
We consider the BS to be equipped with multiple antennas.
To serve the different users, the BS uses a total transmission power $P_{BS}$ that follows the following model~\cite{richter2009energy},
\begin{equation}\label{PBS}
P_{BS}=\theta . \sum_{k=1}^K\sum_{i=1}^{N}P_{BS,k}^i +\varepsilon
\end{equation}
The coefficient $\theta$ captures the power consumption which scales with the radiated power due to amplifier and feeder losses.
The term $\varepsilon$ models the offset of power consumed by the BS regardless of the radiated power due to information processing, battery backup, and cooling.
The BS is powered from a retailer. Assume $\pi$ is the cost of one unit of energy (e.g. the price of $1KWh$) provided by a retailer. We consider just one time slot $\Delta_t$. Hence, the total cost of the procured energy is~\cite{ghazzai2013performance}
\begin{equation}\label{priceBS}
\mathcal{C}_{BS}=\pi. \Delta_t. P_{BS}.
\end{equation}
At the BS side, substituting the expression of the power given by Equation~\eqref{PBS} in Equation~\eqref{priceBS}, we get
\begin{eqnarray}\label{eqn:cost_bs}\nonumber
\mathcal{C}_{BS}&= &\pi. \Delta_t. (\theta .  \sum_{k=1}^K\sum_{i=1}^{N}P_{BS,k}^i +\varepsilon)\\ \nonumber
&=&\pi. \Delta_t. \theta .  \sum_{k=1}^K\sum_{i=1}^{N}P_{BS,k}^i +\pi. \Delta_t. \varepsilon\\
&=&\alpha   \sum_{k=1}^K\sum_{i=1}^{N}P_{BS,k}^i +\varsigma
\end{eqnarray}
where $\alpha=\pi. \Delta_t. \theta$ and $\varsigma=\pi. \Delta_t. \varepsilon$ are parameters to characterize the BSs' power consumption cost. At each user's side, the amount of energy required to receive packets from the BS is (ignoring the acknowledgement)~\cite{vazifehdan2012analytical} $E_{k}^{r}=P_{0}\Delta_t$
where $P_{0}$ is an amount of power used for receiving packets. This power may include the required power for performing the channel estimation.
On the other hand, to offload its data to the BS during time slot $t$, the user consumes~\cite{vazifehdan2012analytical}
$E_{k}^{s}=\Big(P_{0}^{\prime}+{\sum_{i=1}^NP_{k,BS}^i}\Big)\Delta_t$ Joules
where $P_{0}^{\prime}$ is the processing power required prior to sending at each user. Thus, the total transmit and receive power required at each user is $P_{k}^{\textrm{tot}}=P_{0}+P_{0}^{\prime}+{\sum_{i=1}^NP_{k,BS}^i}$.

On the other hand, the cost of the power consumed by each user from its battery is $\mathcal{C}_{k}=\beta_{k}(P_{k}^{\textrm{tot}}-Q_k)$,
where $\beta_{k}$ is a weighting coefficient that captures the attitude of each user whether to rely on its battery or harvesting from the received RF signals. Note that typical numbers for the different parameters used to define the cost functions can be found in~\cite{richter2009energy,ghazzai2013performance,vazifehdan2012analytical}. Since some of these variables are changing over time (e.g., $\pi$) and may change
from one device to another, we instead introduce the variable $\kappa_k=\frac{\beta_k}{\alpha}$  and study its effect. When $\kappa_k$ is very small, the behavior of the system encourages the users to consume power from their batteries first. In the other case, it encourages harvesting from the BS's RF signal.

We start by the case where the BS powers the devices using dedicated subcarriers signals.
Note that this scenario is appropriate for devices with limited hardware capabilities~\cite{incel2011survey,buckley2012novel,MTM} typically used for health and fitness (body sensor devices) or industrial IoTs applications. Hardware restrictions limit devices to only tune and communicate over a small number of channels (e.g., the $N$ subcarriers/channels assigned to each user), but not over a large number of channels to cover all the $K\times N$ channels used by the BS (as in the second system setup presented in Section~\ref{sec:probfor1}).

When each user can only harvest energy from its dedicated RF subcarrier signals, the global problem of jointly minimizing the power utility is formulated as
\begin{subequations}\label{eqn:probform}
\begin{align}
\!\!\!\!\!\displaystyle\min_{\{\rho_k,\{P_{BS,k}^i\}_{i=1}^N, \{P_{k,BS}^i\}_{i=1}^N\}_{k=1}^K} ~~ &~\mathcal{C}_{BS}+\sum_{k=1}^K \mathcal{C}_k \label{eqn:obj}\\
\textrm{s.t.}~ &P_{k}^{\textrm{tot}} -{P}_{k}^{bat}\leq Q_{k},\label{eqn:cons1}\\
&~ R_{BS,k}\geq r_{BS,k}^{th}, \label{eqn:cons2}\\
&~ R_{k,BS} \geq r_{k,BS}^{th},\label{eqn:cons3}\\
&~P_{BS,k}^i\geq 0,~P_{k,BS}^i\geq0,\label{eqn:cons4}\\
&~0\preceq\boldsymbol{\rho}\preceq1\label{eqn:cons5}
\end{align}
\end{subequations}
where $\boldsymbol{\rho}=[\rho_1,...,\rho_K]^T$. Equation~\eqref{eqn:obj} expresses the global objective. Constraint~\eqref{eqn:cons1} controls the total power budget at the $k^{th}$ user, so that it does not exceed the harvested power plus the remaining battery's power. Constraints~\eqref{eqn:cons2} and~\eqref{eqn:cons3} are used to meet the data rates for the downlink and uplink streams, respectively. $r_{BS,k}^{th}$ is the minimum downlink rate threshold that should be achieved by the BS when communicating with user $k$~\cite{Zorba2015}, while $r_{k,BS}^{th}$ is the data rate threshold that should be achieved in the uplink by user $k$. Constraints~\eqref{eqn:cons4} and~\eqref{eqn:cons5} ensure the positivity of the allocated power levels and the splitting ratios.
\begin{proposition}\label{prop1}
Under fixed splitting ratio $\boldsymbol{\rho}$, the optimization problem~\eqref{eqn:probform} is a convex optimization problem.
\end{proposition}
\begin{proof}
  See Appendix~A.
\end{proof}

The original optimization problem~\eqref{eqn:probform} is not convex which makes it hard to find optimal solution via standard optimization tools.  Using Proposition~\ref{prop1}, if we fix $\boldsymbol{\rho}$ to the optimal splitting ratio $\boldsymbol{\rho}^{opt}$, this optimization problem can be formulated as two successive convex optimization problems that can be solved efficiently by feeding the optimal solution of the first optimization to the second optimization. Therefore, we propose to proceed as follows.
We compute the amount of power to be needed in the uplink as a first step to quantify how much power should be harvested by each user. In a second step, we determine the downlink power levels that are to be used at the BS to meet the downlink data rate threshold, as well as the amount of power needed by the users for the uplink communications, as determined in the previous step. Then, we perform an exhaustive search for the optimal splitting ratio $\rho_k^{opt}$ that minimizes the total consumed power.

Next, we determine the amount of energy needed by user $k$ to meet its required uplink data rate, $r_{k,BS}^{th}$.

\subsection{Optimal Uplink Power Allocation}
In the uplink, each user minimizes its transmit power subject to meeting its required data rate. This can be formulated as:
\begin{subequations}\label{eqn:reformulation2}
\begin{align}
\min_{\{P_{k,BS}^i\}_{i=1}^N}  &~~~~~~~~~\mathcal{C}_k, \\
\textrm{s.t.}&~~~~~~~~~ R_{k,BS}\geq r_{k,BS}^{th}\label{eqn:constLagOpt21}
\end{align}
\end{subequations}
The solution to \eqref{eqn:reformulation2} is given by the following lemma
\begin{lemma}\label{lemma1}(The power allocation in the uplink)\\
The optimal power allocation in the uplink for user $k$ is
\begin{equation}\label{eqn:pb}
P_{k,BS}^{i^*}=\bigg[\nu_k-\frac{\sigma_{k,BS}^i}{{|h_{k,BS}^i|}^2}\bigg]_0^+
\end{equation}
where
\begin{equation} \nu_k={\Big({2^{\frac{r_{k,BS}^{th}}{B}}}/({\prod_{j\in\mathcal{U}_k}{|h_{k,BS}^j|}^2/\sigma_{k,BS}^j})\Big)}^{1/{|\mathcal{U}_k|}}
\end{equation}
and $\mathcal{U}_k=\{i|\nu_k-\sigma_{k,BS}^i/{|h_{k,BS}^i|}^2\geq0\}$. $|\mathcal{X}|$ is the cardinality of $\mathcal{X}$ and $[x]_0^+=\max\{0,x\}$.
\end{lemma}
\begin{proof}
  See Appendix~B.
\end{proof}
Having determined the power level, $P_{k,BS}^{i^*}$, user $k$ needs to be able to communicate its data over subcarrier $i$, which allows us to set the variable, $P_{BS,k}^i$, in the optimization problem~\eqref{eqn:probform} to $P_{k,BS}^{i^*}$, and solve for the downlink power variables, $P_{BS,k}^i$.

\subsection{Optimal Downlink Power Allocation}
In the downlink, the BS aims to find the optimal power level that it has to transmit in order to meet each user $k$'s downlink data rate, $r_{BS,k}^{th}$, and to be able to power each user $k$ with enough power to allow it to meet its required uplink rate, $r_{k,BS}^{th}$. This is derived with respect to the power utility defined earlier. In this first system setup, each user can only harvest from its RF signal subcarriers. Given the uplink power needed at each user, which is determined by Equation~\eqref{eqn:pb}, the optimization problem at the BS is then formulated as

\begin{subequations}\label{eqn:reformulation0}
\begin{align}
\!\!\!\!\!\!\!\!\min_{\big\{\{P_{BS,k}^i\}_{i=1}^N\big\}_{k=1}^K}  &~\tilde{\varsigma}+\sum_{k=1}^K\sum_{i=1}^N\tilde{\alpha}_k^iP_{BS,k}^i, \label{eqn:objsub}\\
\textrm{s.t.} &~ R_{BS,k}\geq r_{BS,k}^{th},~k\in[1..K],\label{eqn:constLagOpt20}\\
&~ \sum_{i=1}^N P_{BS,k}^i {|h_{BS,k}^i|}^2\geq P_{k}^{th},~k\in[1..K],\label{eqn:constLagOpt30}
\end{align}
\end{subequations}
where $\tilde{\varsigma}=\varsigma+\sum_{k=1}^K\beta_k(\sum_{i=1}^{N}(P_{k,BS}^{i^*}-\eta\rho_k\sigma_{BS,k}^i)+P_{0}+P_{0}^{\prime})$ and $\tilde{\alpha}_k^i=\alpha-\beta\eta\rho_k{|h_{BS,k}^i|}^2$. The quantity $P_{k}^{th}$ represents a power threshold that we deduce from the constraint~\eqref{eqn:cons1}. It depends on the amount of power needed for achieving the required uplink data rate, the amount of power available in the battery, the splitting ratio, the conversion efficiency, and the noise power, and is expressed as
\begin{equation}\label{eqn:newVariables}
P_{k}^{th}=\frac{\sum_{i=1}^NP_{k,BS}^{i^*}+P_{0}+P_{0}^{\prime}-{P}_{k}^{bat}}{\eta \rho_k}-\sum_{i=1}^N\sigma_{k,BS}^i
\end{equation}
Minimizing the affine objective function given by Equation~\eqref{eqn:objsub} is equivalent to minimizing the linear quantity $\sum_{k=1}^K\sum_{i=1}^N\tilde{\alpha}_k^iP_{BS,k}^i$. Hence, we could re-write the problem for each $k\in[1..K]$ as follows
\begin{subequations}\label{eqn:reformulation}
\begin{align}
\min_{\{P_{BS,k}^i\}_{i=1}^N}  &~~~~~~~~~\sum_{i=1}^N\tilde{\alpha}_k^iP_{BS,k}^i, \\
\textrm{s.t.} &~~~~~~~~~ R_{BS,k}\geq r_{BS,k}^{th},\label{eqn:constLagOpt2}\\
&~~~~~~~~~ \sum_{i=1}^N P_{BS,k}^i {|h_{BS,k}^i|}^2\geq P_{k}^{th}\label{eqn:constLagOpt3}
\end{align}
\end{subequations}
The optimal per-user per-subcarrier downlink power allocation is given by the following theorem.

\begin{theorem}\label{theo:1}
The solution to~\eqref{eqn:reformulation} above is
\begin{equation}\label{eqn:pa}
P_{BS,k}^{i^*}=\bigg[\frac{\lambda_{k}}{\tilde{\alpha}_k^i-\psi_{k}{|h_{BS,k}^i|}^2}-\frac{\sigma_{BS,k}^i}{(1-\rho_k){|h_{BS,k}^i|}^2}\bigg]_0^+,
\end{equation}
for $i\in[1..N]$ and $k\in[1..K]$, where
\begin{eqnarray}\label{eqn:lambda0}\nonumber
\bullet \;\lambda_{k}=2^{\footnotesize{{\frac{r_{BS,k}^{th}}{B|\mathcal{S}_k|} - \frac{1}{|\mathcal{S}_k|} \log_2\Big(\displaystyle{\prod_{i\in \mathcal{S}_k}}\frac{(1-\rho_k){|h_{BS,k}^i|}^2}{(\tilde{\alpha}_k^i-\psi_{k}{|h_{BS,k}^i|}^2)\sigma_{BS,k}^i}\Big)}}}
\end{eqnarray}
$\bullet$ $\mathcal{S}_k=\{i|\lambda_{k}/(\tilde{\alpha}_k^i-\psi_{k}{|h_{BS,k}^i|}^2)>\sigma_{BS,k}^i/(1-\rho_k){|h_{BS,k}^i|}^2\}$, \\
 $\bullet$ $\psi_k$ is the zero of the function
\begin{eqnarray}\label{eqn:fct}\nonumber
 f(x)&=&\frac{2^{r_{BS,k}^{th}/{B|\mathcal{S}_k|}}\displaystyle{\sum_{i\in \mathcal{S}_k}}\frac{{|h_{BS,k}^i|}^2}{\tilde{\alpha}_k^i-x{|h_{BS,k}^i|}^2}}{{\Big(\displaystyle{\prod_{i\in \mathcal{S}_k}}\frac{(1-\rho_k){|h_{BS,k}^i|}^2}{(\tilde{\alpha}_k^i-x{|h_{BS,k}^i|}^2)\sigma_{BS,k}^i}\Big)}^{\frac{1}{|\mathcal{S}_k|}}}\\
 &-&P_{k}^{th} -\sum_{i\in \mathcal{S}_k}\frac{\sigma_{BS,k}^i}{1-\rho_k}
\end{eqnarray}
\end{theorem}
\begin{proof}
  See Appendix~C.
\end{proof}
The theorem provides the optimal power levels in the downlink, but requires to find the zero of the function $f$. In what follows, we first prove the existence of a zero, and then present a technique for finding it.

To examine the monotony of $f$, we take the derivative over $x$. It follows that the sign of $f^{\prime}$ is the sign of
\begin{equation}
\displaystyle{\sum_{i\in \mathcal{S}_k}}\frac{{|h_{BS,k}^i|}^4}{{(\tilde{\alpha}_k^i-x{|h_{BS,k}^i|}^2)}^2}- \frac{1}{{|\mathcal{S}_k|}}{{\Big(\displaystyle{\sum_{i\in \mathcal{S}_k}}\frac{{|h_{BS,k}^i|}^2}{{\tilde{\alpha}_k^i-x{|h_{BS,k}^i|}^2}}}\Big)^2}
\end{equation}
Recalling Cauchy-Schwartz inequality, ${\big(\sum_{i=1}^N|x_i|\big)}^2\leq N\sum_{i=1}^N{|x_i|}^2$, we conclude that $f$ is a non-decreasing function.
Since $f$ presents some points of singularity, it is sufficient to prove the existence of an interval where $f$ is continuous and the signs of $f$ at the interval boundaries are opposite. Let $\phi_0=0$ and $\phi_i=\tilde{\alpha}_k^i/{|h_{BS,k}^i|}^2$. Without loss of generality, we assume that $\phi_0<\phi_1<\phi_2<..$.
Now, the $\displaystyle\lim_{x \rightarrow \phi_{0}^+}f(x)$ is finite, but could be positive or negative, depending on the numerical values of the different systems parameters. On the other hand, $\displaystyle\lim_{x \rightarrow \phi_{1}^-}f(x)$ goes to $+\infty$. Hence, if $\displaystyle\lim_{x \rightarrow \phi_{0}^+}f(x)$ is negative, then there is a zero of $f$ in this interval. Otherwise, we check the next interval. Note that the search is restricted to the intervals  $\mathcal{I}_l=]\phi_l,\phi_{l+1}[$ with $l\in\{0,2,4,6,..\}$ to ensure the non negativity of $\Big(\displaystyle{\prod_{i\in \mathcal{S}_k}}\frac{(1-\rho_k){|h_{BS,k}^i|}^2}{(\tilde{\alpha}_k^i-\psi_{k}{|h_{BS,k}^i|}^2)\sigma_{BS,k}^i}\Big)$ in order to satisfy Equation~\eqref{eqn:lambda0}.
Since $\displaystyle\lim_{x \rightarrow \phi_{2}^+}f(x)=-\infty$ and $\displaystyle\lim_{x \rightarrow \phi_{3}^-}f(x)=+\infty$, then there is a zero in this interval.
Given the presence of the intervals of singularities, we propose to use the bisection method to find the zero of $f$, which essentially searches for the zero of $f$ incrementally in each interval. A check of the sign of their product $\Big(\displaystyle{\lim_{x \rightarrow\phi_{l}^+}} f(x).\displaystyle{\lim_{x \rightarrow\phi_{l+1}^-}} f(x)\Big)$, $l=0,2,..$ will be sufficient to decide the search.

So far, we have analytically derived the optimal power levels that need to be allocated by the BS to achieve its required downlink data rates (i.e., from the BS to each user), as well as to allow each user to achieve its uplink data rates (i.e., from the users to the BS) by giving it enough power to harvest and use for uplink communication.
We now propose an efficient and practical algorithm that finds these optimal power levels.
This algorithm is based on the theory developed in this section.

\subsection{An Efficient Algorithm for Solving the Joint Uplink and Downlink Optimization Formulation}
After deriving the optimal power for the uplink and downlink communications in the two previous subsections, we now use these results to propose our Algorithm \ref{alg:1}. Remember that we kept the splitting ratio $\rho_k$ as a design parameter in a first step to make the problem convex. This parameter could be optimized at this level using an exhaustive search method to derive the optimal splitting ratio $\rho_k^{opt}$ that minimizes the total power consumption.
Also, note that this parameter could be optimized for every time slot as it depends on the channels' quality.
\begin{algorithm}
\caption{Joint Power Allocation}
\begin{algorithmic}[1]
\REQUIRE $\{r_{BS,k}^{th}\}_{k=1}^K$, $\{r_{k,BS}^{th}\}_{k=1}^K$, ${|h_{BS,k}^i|}^2$, ${|h_{k,BS}^i|}^2$, $N$, and $\{\bar{P}_{k}\}_{k=1}^K$
\FOR{$k=1:K$}
\STATE Compute $\{P_{k,BS}^j\}_{j=1}^N$ using \eqref{eqn:pb}
\FOR{$\rho_k=0:1$}
    \STATE Find $\psi_{k}$ using bisection method applied to~\eqref{eqn:fct}
    \STATE Compute $\lambda_{k}$ using~\eqref{eqn:lambda0}
    \STATE Compute $\{P_{BS,k}^i\}_{i=1}^N$ using~\eqref{eqn:pa}
\ENDFOR
\STATE Find $\rho_k^{opt}$
\ENDFOR
\RETURN $\{P_{BS,k}^i\}_{i=1}^N$, $\{P_{k,BS}^i\}_{i=1}^N$, $\rho_k^{opt}$ $\forall k\in [1..K]$
\end{algorithmic}\label{alg:1}
\end{algorithm}
The numerical evaluations of our optimization are provided in Section~\ref{sec:simulationResults}.
It is worth mentioning that our framework considers that the BS has enough processing capabilities to handle the computation complexity of Algorithm~\ref{alg:1}.

While our approach allows users to communicate with the BS using harvested energy, additional power savings can further be achieved when users have sufficient hardware capabilities. For instance, when equipped with appropriate hardware, having each user also harvest from the subcarriers used by any other user will result in harvesting more energy. Moreover, a user can also harvest from any other RF signals sent by any neighboring BSs, though in this case there is no guarantee that the BS is transmitting with the least amount of power.
The essence of our proposed optimization framework is to guarantee that the harvested amount of energy is enough for the users to meet their required rates in the uplink, and to do so with the least possible amount of power the BS will have to consume. Now that we considered the case of harvesting from dedicated RF signals only, in the next section, we consider the case of minimizing the BS's power consumption while assuming that users are equipped with sufficient hardware capability that allows them to harvest energy from their dedicated RF signals as well as the other users' signals. That is, each user can harvest energy from any ambient RF signal sent by its BS, whether destined to it or to other users serviced by the same BS.

\section{Hybrid Dedicated and Ambient RF Signal Based Energy Harvesting}
\label{sec:probfor1}
In the previous section, we considered the case where a user can harvest energy only from the RF signals that are intended for it by the BS.
However, a user can still receive RF signals, though as interference, even when the signals are not meant to be sent to it.
Therefore, a more general setup we consider here is to assume that a user can harvest energy not only from its intended RF signals, but also from all other ambient RF signals sent by the BS to any user.
We anticipate that by doing so, the overall amount of energy to be consumed by the system will be reduced.
In this section, we solve the power allocation optimization problem for this general setup.

The problem formulation remains the same as in~\eqref{eqn:probform} except that $Q_{k}$ in the Constraint~\eqref{eqn:cons1} needs to be replaced by
\begin{equation}\label{eqn:harvPow}
Q_{k}=\eta\rho_k\sum_{l=1}^K\sum_{i=1}^NP_{BS,l}^i{|h_{BS,l}^{i,k}|}^2+\sigma_{BS,l}^i,
\end{equation}
where $h_{BS,l}^{i,k}$ is the downlink channel impulse between the BS and the $k^{th}$ user that corresponds to the $(l-1)\times N+i$ subcarrier that is normally allocated for the communication between the BS and user $l$.

We use the same steps as in the previous section for solving this problem. We start by computing the amount of power needed by each user in the uplink. Here, the optimal power over each subcarrier is the same as the one derived in the previous section, given by Equation~\eqref{eqn:pb}. In the downlink, the BS accounts for power needed by each user $k$ so that the harvested amount of power satisfies $Q_{k}+{P}_k^{bat}\geq \sum_{i}^NP_{k,BS}^{i^*}$. In addition, the BS should meet the downlink rate $r_{BS,k}^{th}$.

Let $\boldsymbol{P}_{BS,k}=\big[P_{BS,k}^{1},...,P_{BS,k}^{N}\big]$ be the vector containing the power levels that the BS allocates for communication with user $k$, and $\boldsymbol{P}={[\boldsymbol{P}_{BS,1},...,\boldsymbol{P}_{BS,K}]}^T$ be the vector containing the power levels used for the communication with all the users. Also, let $\boldsymbol{h}_{BS,k}^l=\Big[{|h_{BS,k}^{1,l}|}^2,...,{|h_{BS,k}^{N,l}|}^2\Big]$ and $\boldsymbol{h}_l={[\boldsymbol{h}_{BS,1}^l,...,\boldsymbol{h}_{BS,K}^l]}^T$.
Hence, the optimal downlink power allocation for each user is the solution to
\begin{subequations}\label{eqn:reformulation11}
\begin{align}
\min_{\{\{P_{BS,k}^i\}_{i=1}^N\}_{k=1}^K}  &~\varsigma^{\prime}+\boldsymbol{\tilde{\alpha}}^T\boldsymbol{P}, \\
\textrm{s.t.}&R_{BS,k}\geq r_{BS,k}^{th},\\
&~ \boldsymbol{P}^T\boldsymbol{h}_k\geq P_{k}^{th}\label{eqn:constLagOpt31}\\
&~ \boldsymbol{P}\succeq 0, \label{eqn:prob2}
\end{align}
\end{subequations}
where $\varsigma^{\prime}=\varsigma+\sum_{k=1}^K\beta_k\Big[\sum_{i=1}^NP_{k,BS}^{i^*}+P_{0}+P_{0}^{\prime}-\eta\rho_k\sum_{l=1}^K\sum_{j=1}^N\sigma_{BS,l}^j\Big]$, $\boldsymbol{\tilde{\alpha}}=\alpha\boldsymbol{1}+\sum_{k=1}^K\beta_k\eta\rho_k\boldsymbol{h}_k$ and 
$P_{k}^{th}=\frac{\sum_{i=1}^NP_{k,BS}^{i^*}+P_{0}+P_{0}^{\prime}-\bar{P}_{k}}{\eta \rho_k}-\sum_{l=1}^K\sum_{i=1}^N\sigma_{l,BS}^i$.

Our objective is to derive the optimal downlink power for each user such that the total cost is minimized. The following theorem gives the optimal power allocation in this scenario.
\begin{theorem}\label{theo:2}
The solution to the optimization problem~\eqref{eqn:reformulation11} is
\begin{equation}\label{eqn:pa22}
\!\!P_{BS,k}^{i^*}={\bigg[\frac{\lambda_{k}^{\prime}}{\tilde{\alpha}_k^i-\sum_{l=1}^K\psi_{l}{|h_{BS,k}^{i,l}|}^2}-\frac{\sigma_{BS,k}^i}{(1-\rho_k){|h_{BS,k}^{i}|}^2}\bigg]}_0^+
\end{equation}
where $\lambda_k^{\prime}$ and $\psi_{k}$ are the K.K.T. multipliers to be specified later.
\end{theorem}
\begin{proof}
  See Appendix~D.
\end{proof}

Note that with comparison to the power level that we have derived in the other section (given by Equation~\eqref{eqn:pa}), the expression accounts for the interference channels between users.
In total, we have $2K$ K.K.T. multipliers, $\{\lambda_{k}^{\prime},\psi_{k}\}_{k=1}^K$, that we compute by replacing $P_{BS,k}^{i^*}$ in the K.K.T. conditions. The following lemma provides a characterization of $\{\lambda_{k}^{\prime},\psi_{k}\}_{k=1}^K$.
\begin{lemma}\label{lemma2}
The expression of $\lambda_k^{\prime}$ is given by
\begin{equation}\label{eqn:lambda011}
\lambda_{k}^{\prime}=2^{\frac{r_{BS,k}^{th}}{B|\mathcal{C}_k|}} {\bigg(\displaystyle{\prod_{i\in \mathcal{C}_k}}\frac{\gamma_{k}^i}{\tilde{\alpha}_k^i-\sum_{l=1}^{K}\psi_{l}{|h_{BS,k}^{i,l}|}^2}\bigg)}^{- \frac{1}{|\mathcal{C}_k|}},
\end{equation}
where $\mathcal{C}_k=\{i|\lambda_{k}^{\prime}/(\tilde{\alpha}_k^i-\sum_{l=1}^K\psi_{l}{|h_{BS,k}^{i,l}|}^2)\geq1/\gamma_{k}^i\}$ and $\boldsymbol{\psi}=[\psi_{1}...\psi_K]^T$ is the zero of the functions
\begin{eqnarray}\label{eqn:lambda111d11}\nonumber
f_k(\boldsymbol{x})\!\!\!\!&=&\!\!\!\!\frac{2^{\frac{r_{BS,k}^{th}}{B|\mathcal{C}_k|}}\sum_{l=1}^{K}\sum_{i\in \mathcal{C}_{l}} \frac{{|h_{BS,l}^{i,k}|}^2}{\tilde{\alpha}_l^i-\sum_{m=1}^{K}x_{m}{|h_{BS,l}^{i,m}|}^2}}{{\bigg(\displaystyle{\prod_{i\in \mathcal{C}_k}}\frac{\gamma_{k}^i}{\tilde{\alpha}_k^i-\sum_{l=1}^{K}x_{l}{|h_{BS,k}^{i,l}|}^2}\bigg)}^{\frac{1}{|\mathcal{C}_k|} }} \\
&-&\sum_{l=1}^K\sum_{i\in \mathcal{C}_{l}}\frac{{|h_{BS,l}^{i,k}|}^2}{\gamma_{l}^i}-P_{k}^{th}
\end{eqnarray}
\end{lemma}
\begin{proof}
  See Appendix~E.
\end{proof}
Note that in Equation~\eqref{eqn:lambda111d11}, we have $K$ unknowns $\{x_{l}\}_{l=1}^K$. But since we have $K$ equations, theoretically, we could solve for $\{\psi_{l}\}_{l=1}^K$.
To prove the existence of a zero of the function $\boldsymbol{f}$ defined as $\boldsymbol{f}(\boldsymbol{\psi})=[f_1(\boldsymbol{\psi}) ... f_K(\boldsymbol{\psi})]^T$ where $f_k(\boldsymbol{\psi})$ is given by Equation~\eqref{eqn:lambda111d11}, we proceed similarly to the proof in the case of the single variable by considering $\psi_j$ to be variable and fixing the other $K-1$ variables. We end up with the previous scenario where the bounds of the search intervals are $\phi_0=0$ and $\phi_{n}=\tilde{\alpha}_k^i/{|h_{BS,k}^{i,j}|}^2-\displaystyle{\sum_{l \in\mathcal{C}_k,l\neq j}}\psi_l{|h_{BS,k}^{i,l}|}^2/{|h_{BS,k}^{i,j}|}$ for $n\in\mathcal{C}_k$ and we restrict $\psi_k$ to be positive. Then a sign check of the limit of the continuous function $f_k$ in the bounds of the intervals is sufficient to prove that a zero exists.

Deriving a closed-form expression of $\boldsymbol{\psi}$ is not possible due to function nonlinearity. However, it could be derived iteratively using the Newton method where, at each iteration, \begin{equation}\label{eqn:new}
  \boldsymbol{\psi}^{n+1}=\boldsymbol{\psi}^{n}-{\big(\nabla \boldsymbol{f}(\boldsymbol{\psi}^{n})\big)}^{-1}\boldsymbol{f}(\boldsymbol{\psi}^{n})
\end{equation}
where $\big(\nabla \boldsymbol{f}(\boldsymbol{\psi})\big)_{kj}=\frac{ \partial f_k(\boldsymbol{\psi})}{\partial \psi_j} ={\xi}^{\prime}(\boldsymbol{\psi})\zeta(\boldsymbol{\psi})+\xi(\boldsymbol{\psi}){\zeta}^{\prime}(\boldsymbol{\psi})$,
with
\begin{subequations}\label{eqn:1}
\begin{align}\label{hdhhd}
\xi(\boldsymbol{\psi})   &= \frac{2^{\frac{r_{BS,k}^{th}}{B|\mathcal{C}_k|}}}{{\bigg(\displaystyle{\prod_{i\in \mathcal{C}_k}}\frac{\gamma_{i,k}}{\tilde{\alpha}_k^i-\sum_{l=1}^{K}\psi_{l}{|h_{BS,k}^{i,l}|}^2}\bigg)}^{\frac{1}{|\mathcal{C}_k|} }}\\
\zeta(\boldsymbol{\psi}) &= \sum_{l=1}^{K}\sum_{i\in \mathcal{C}_{l}} \frac{{|h_{BS,l}^{i,k}|}^2}{\tilde{\alpha}_l^i-\sum_{m=1}^{K}\psi_{m}{|h_{BS,l}^{i,m}|}^2} \\
\xi^{\prime}(\boldsymbol{\psi})&=-\frac{1}{|\mathcal{C}_k|}\frac{2^{\frac{r_{BS,k}^{th}}{B|\mathcal{C}_k|}} \displaystyle{ \sum_{i\in\mathcal{C}_k} } \frac{{|h_{BS,k}^{i,j}|}^2}{\tilde{\alpha}_k^i-\sum_{m=1}^K\psi_m{|h_{BS,k}^{i,m}|}^2 }  }{    {{\bigg(\displaystyle{\prod_{i\in \mathcal{C}_k}}\frac{\gamma_{k}^i}{\tilde{\alpha}_k^i-\sum_{l=1}^{K}\psi_{l}{|h_{BS,k}^{i,l}|}^2}\bigg)}^{\frac{1}{|\mathcal{C}_k|} }}}\\
{\zeta}^{\prime}(\boldsymbol{\psi})&=\sum_{l=1}^{K}\sum_{i\in \mathcal{C}_{l}} \frac{{|h_{BS,l}^{i,k}|}^2{|h_{BS,l}^{i,j}|}^2}{{\big(\tilde{\alpha}_l^i-\sum_{m=1}^{K}\psi_{m}{|h_{BS,l}^{i,m}|}^2\big)}^2}
\end{align}\end{subequations}
With the use of the Newton method for deriving $\boldsymbol{\psi}$, there is a tradeoff between precision and computational complexity. An accuracy $\epsilon$ of $10^{-5}$ is sufficient to get the Newton method converge in a relatively small computational time.

\section{Numerical Evaluation}
\label{sec:simulationResults}
We consider that the BS is placed at the center of a cell with radius $d_0=1$ and that the users are uniformly distributed within the cell. The fading of the channels is modeled as Rayleigh with mean $\sqrt{{[d_0/d_k]}^{\alpha}}$ where $\alpha$ is the pathloss exponent set to $3$, and $d_k$ is the normalized distance between the mobile user $k$ and the BS and is generated randomly between 0 and 1.
At each user, the energy harvesting conversion efficiency $\eta$ is chosen to be equal to $0.8$ while the noise power density is taken equal to $N_0= -174~dBm$ as in~\cite{fu2017power}. The bandwidth of each sub-band is $15kHz$. The number of subcarriers used to communicate with each user, $N$, is taken equal to $10$ unless otherwise specified. Since, the processing time for sending and receiving packets is negligible compared to the transmission power, then we set $P_{0}=P_{0}^{\prime}\approx 0$.

A key parameter in RF energy harvesting systems based on power splitter is the splitting ratio $\rho_k$. We first study its impact on the total power in the system when setting the battery power to zero. Consider the case where the BS communicates with one user, i.e., $K=1$, and we take $\alpha=\beta=1$ and plot the total power in the system as a function of the splitting ratio $\rho_k$ in Fig.~\ref{fig:rho}. First, observe that the required power varies as a function of the splitting ratio, as well as the SNR of the uplink and downlink channels. It decreases with the increase of $\rho_k$ up to $\rho_k^{opt}$, and it starts to increase as $\rho_k$ goes beyond $\rho_k^{opt}$. When $\rho_k<\rho_k^{opt}$, more power is needed to be harvested to meet the user's rate threshold, $r_{k,BS}^{th}$.
On the other hand, if $\rho_k>\rho_k^{opt}$, then the user's needed power is met while the BS needs to increase its transmission power in order to meet $r_{BS,k}^{th}$.
Hence, the splitting ratio strikes a balance between the amount of harvested power and the needed power to meet the data rate. Second, we investigate the effect of the uplink channels' SNR on the total power. As the uplink channels' SNR becomes stronger, less power consumption is needed. Third, for a strong downlink channels' SNR (chosen to be $10$ dB), the optimal splitting ratio $\rho_k^{opt}$ becomes closer to $0$ as the uplink channels' SNR increases. This is because as the channels' SNR becomes stronger, less power is needed to be harvested (proportional to the splitting ratio).
\begin{figure}
\centering{
\includegraphics[width=.9\columnwidth]{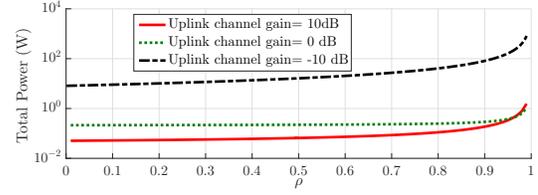}
\caption{The sum power function of the splitting ratio $\rho$. The system parameters are as follows: the number of subcarriers $N=5$, the downlink channel SNR= 10dB, and $r_{BS,k}^{th}=15Kbit/s$ and $r_{k,BS}^{th}=30Kbit/s$.}
\label{fig:rho}}
\end{figure}

In Fig.~\ref{fig:rho_opt}, we plot the optimal splitting ratio $\rho_{k}^{opt}$, found using Algorithm~\ref{alg:1}, as a function of the uplink channels' SNR. Observe that the optimal splitting ratio decreases as the uplink channels' SNR increases. This confirms the result shown in Fig.~\ref{fig:rho} since if the uplink channel quality increases, less power is needed to meet the user's rate requirement and therefore less amount of harvested energy is required. If the downlink channels' quality becomes worse, the optimal splitting ratio decreases in the low SNR regime. In fact, more power should be dedicated to meet the downlink rate requirement $r_{BS,k}^{th}$. This is equivalent to the increase of the portion dedicated for decoding (i.e., $1-\rho_k$).

\begin{figure}
\centering{
\includegraphics[width=.9\columnwidth]{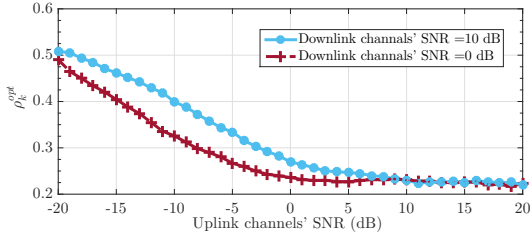}
\caption{The optimal splitting ratio $\rho_k$ as a function of the uplink channels' SNR. The system parameters are as follows: the number of subcarriers $N=5$, the downlink channel SNR= 10dB, $r_{BS,k}^{th}=300Kbit/s$ and $r_{k,BS}^{th}=150Kbit/s$.}
\label{fig:rho_opt}}
\end{figure}
In Fig.~\ref{fig:SumPower_SNR}, we investigate the effect of the downlink channels' SNR for a fixed $\rho_k$ and for the optimal $\rho_k^{opt}$ under different channels' SNR values. We plot the total power consumption while varying the downlink channels' SNR for a fixed splitting ratio, $\rho=0.5$. Note that the total power consumption decreases as the downlink channels' SNR increases. As the downlink channel gains become stronger, less power is required to achieve the required data rate, $r_{BS,k}^{th}$. Also, as the uplink channels' SNR increases, the total power consumption decreases. This is because the user needs less power to achieve its required data rate, $r_{k,BS}^{th}$. Furthermore, we plot in the same figure the sum power using the optimal splitting ratio $\rho_k^{opt}$. Note that an additional gain in the power consumption is obtained with the optimal splitting ratio. However, from a practical perspective, this comes at the expense of a more sophisticated circuitry design, as the optimal ratio needs to be found at each time slot, depending on the channels' gains.
\begin{figure}
\centering{
\includegraphics[width=.9\columnwidth]{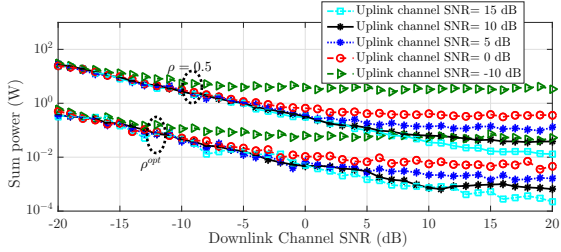}
\caption{The sum power as a function of the downlink channel SNR between the BS and one user. The system parameters are as follows: the number of subcarriers $N=5$, $r_{BS,k}^{th}=15Kbit/s$, and $r_{k,BS}^{th}=30Kbit/s$.}
\label{fig:SumPower_SNR}}
\end{figure}

Having studied the impact of the splitting ratio in the case of one single user, we now assess the performance of our framework by considering the following metrics: the total power cost (utility cost) and the system lifetime. Note that in the case of not harvesting from the received RF signals and if the power at the users' batteries is not sufficient to meet their data rates, an outage performance occurs. That is, the users can not offload their data.
On the other hand, the system lifetime is usually maximal when the harvesting capability enabled.

We assume that $\kappa_k=\kappa=\beta/\alpha$ and to ensure the non negativity of $\tilde{\alpha}_k^i$ as well as $\kappa$, the values of $\kappa$ should be picked in the interval $[0..\displaystyle{\min_{k,i\in\mathcal{C}_k}}(1/(\eta\rho{|h_{BS,k}^i|}^2))]$. In Fig.~\ref{fig:Cost_powerkappa}, we plot the power utility as a function of $\kappa$. Observe that relaying on harvesting results on a higher cost compared to the case when relaying solely on the batteries at the different users. Furthermore, we observe that the utility cost for both systems enhances as $\kappa$ increases. In fact, as $\kappa$ increases, $\beta$ increases as well, and hence, the cost relative to the power used from the battery augments and affects the total cost. %
\begin{figure}
\centering{
\includegraphics[width=.9\columnwidth]{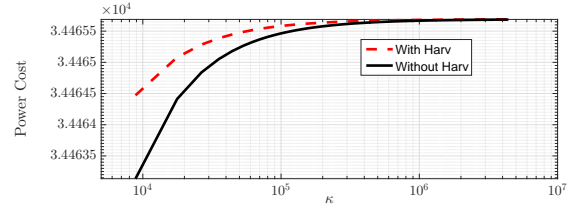}
\caption{The power utility as a function of $\kappa$. The system parameters are as follows: $K=50$ with $N=5$ subcarriers in the uplink and $N$ in the downlink, the uplink and downlink channels SNR= 10dB, $r_{BS,k}^{th}=15Kbit/s$ and $r_{k,BS}^{th}=30Kbit/s$.}
\label{fig:Cost_powerkappa}}
\end{figure}

The second performance metric that we consider is the lifetime. While in wireless sensors networks problems, the lifetime is often defined as the number of transmission time slots from the deployment up to when the first sensor's battery dies~\cite{Yunxia2005}, we define it here as the overage time until the users' battery dies. We consider that the users have an equal initial amount of power ${P}_k^{bat}$. In Fig.~\ref{fig:lifetime}, we plot the lifetime as a function of the users' number.
We define $\epsilon$ as the portion of power used from the battery, while (1-$\epsilon$) is the portion of power harvested from the BS's signal. Hence, $\epsilon=1$ corresponds to the use of the total power needs from the battery.
First, we notice that regardless of the number of users in the system, as long as we harvest a portion of the power, we achieve a higher lifetime compared to the case when we solely rely on the user's battery. Second, as the portion of the power taken from the battery decreases, a higher system lifetime is achieved. When $\epsilon$ tends to zero, the lifetime goes to infinity. This is because almost all the power is harvested from the RF signals. Third, normally the placement of the users themselves in the cell affects the performance. However, the figure shows that the curves are almost flat as a function of the number of users. This is because, we average over a large number of users' placement. Last, recalling Figure~\ref{fig:Cost_powerkappa}, a tradeoff between the cost and the lifetime should be struck.
\begin{figure}
\centering{
\includegraphics[width=.9\columnwidth]{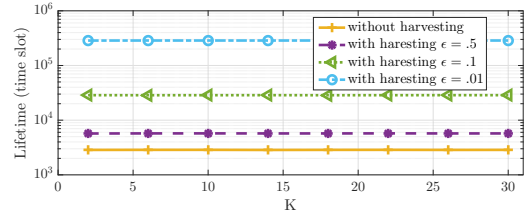}
\caption{The system lifetime as a function of the uniformly deployed sensors. The system parameters are as follows: One user with $N=5$ subcarriers in the uplink and $N$ in the downlink, the uplink and downlink channels SNR= 10dB, and $r_{BS,k}^{th}=100Kbit/s$ and $r_{k,BS}^{th}=1Mbit/s$.}
\label{fig:lifetime}}
\end{figure}

Now, we look at the BS power allocation as a function of the channels' variations.
When accounting for harvesting from the signals intended to the other users, we anticipate to achieve further power savings at the BS. Hence, we compare the total power allocated for powering and communicating with the different users in the two system setups: when each user harvests power only from its intended signals and when, in addition to that, each user harvests energy from the signals dedicated to the other users. In Fig.~\ref{fig:CompScene}, we plot the total power used by the BS for the two setups and for different number of users as a function of the downlink channels' SNR. Note that the users are assumed to have equal average downlink channels' SNR; however, the results are still valid for a more general system. First, remark that an increase in the number of users results in an augmentation in the needed transmission power at the BS. On the other hand, the second system setup allows to achieve less power consumption when compared with the first one. This is because accounting for the received interference at each user, which leads to increasing the amount of the harvested power, decreases the total required power at the BS to serve the users. However, this substantial gain comes at the expense of the prior knowledge of all the channel gains between the BS and the users. Fortunately, this is required in spectrum assignment process.
 \begin{figure}
\centering{
\includegraphics[width=.9\columnwidth]{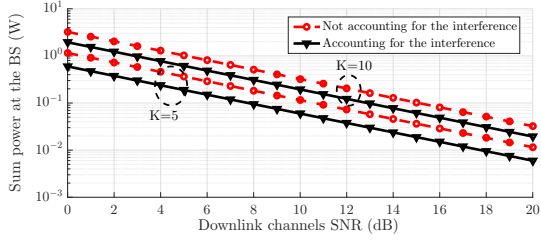}
\caption{Comparison of the BS's sum power allocated to power the users and achieve their downlink rate for the two scenarios: each user harvests only from the subcarriers used for its communication and when each user harvests also from the interference. The system parameters: the uplink channels' SNR $10$ dB and $r_{BS,k}^{th}15kHz$ and $r_{k,BS}^{th}=30Kbit/s$.}
\label{fig:CompScene}}
\end{figure}

\section{Conclusion}
\label{sec:conclusion}
This paper investigates the optimal power allocation for a multiuser multicarrier communication system composed of a base station and mobile users.
We solved the optimal power allocation at the base station to enable data communication as well as powering the users using RF energy harvesting.
We studied the tradeoff between the power cost and system lifetime. Our performance analysis shows that the power consumption gain takes advantage of the variability of channels' gains, the splitting ratio, and the number of subcarriers.
Energy harvesting capability increases the network lifetime, however, this comes with the expense of a higher power cost.

\section*{Appendix}
\subsection{Proof of Proposition~\ref{prop1}}
\begin{proof}
In general, the optimization problem~\eqref{eqn:probform} is not convex. While the objective and the first constraint are affine functions, the two constraints~\eqref{eqn:cons2} and~\eqref{eqn:cons3}, as function of $\boldsymbol{\rho}$, are not convex. By fixing its value, $\boldsymbol{\rho}$ is no longer an optimization parameter, and hence, the problem becomes convex.
\end{proof}

\subsection{Proof of lemma~\ref{prop1}}
\begin{proof}
The solution to~\eqref{eqn:reformulation2} is straightforwardly derived by minimizing the Lagrangian dual function. It is the classical water filling~\cite{Boyd:2004:CO:993483}.
\end{proof}

\subsection{Proof of Theorem~\ref{theo:1}}
\begin{proof}
Since the optimization problem~\eqref{eqn:reformulation} is convex, we consider the dual problem using the Karush-Kuhn-Tucker (K.K.T) conditions. The Lagrangian can be written as
\begin{eqnarray}\label{eqn:lagr}\nonumber
\mathcal{L}_k\big(\{P_{BS,k}^i\}_{i=1}^N\big)&=&\sum_{i=1}^N\tilde{\alpha}_k^iP_{BS,k}^i-\lambda_{k}(R_{BS,k}-r_{BS,k}^{th})\\
&-&\psi_{k}\big(\sum_{i=1}^NP_{BS,k}^i{|h_{BS,k}^i|}^2 - P_{k}^{th}\big),
\end{eqnarray}
where $\lambda_{k}$ and $\psi_{k}$ are the K.K.T. multipliers~\cite{boyd2004}. By taking the derivative of the Lagrangian $\mathcal{L}_k$ over $P_{BS,k}^i$ and set it to zero, we get
\begin{eqnarray}\label{eqn:constLagOpt1}\nonumber
\tilde{\alpha}_k^i&-&\frac{\lambda_{k}(1-\rho_k){|h_{BS,k}^i|}^2/\sigma_{BS,k}^i}{(1+P_{BS,k}^i(1-\rho_k){|h_{BS,k}^i|}^2/\sigma_{BS,k}^i)\log(2)}\\
&-&\psi_{k}{|h_{BS,k}^i|}^2=0
\end{eqnarray}
From Equation \eqref{eqn:constLagOpt1} and by letting $\lambda_k^{\prime}=\lambda_k/\log(2)$, the power level $P_{BS,k}^i$ is, therefore, expressed as
\begin{equation}
P_{BS,k}^i=\frac{\lambda_{k}^{\prime}}{\tilde{\alpha}_k^i-\psi_{k}{|h_{BS,k}^i|}^2}-\frac{\sigma_{BS,k}^i}{(1-\rho_k){|h_{BS,k}^i|}^2}.
\end{equation}
Then, we restrict the power to be positive or null to ensure the positivity of the power levels. To find the Lagrange multipliers $\lambda_k^{\prime}$ and $\psi_{k}$, we rely on the K.K.T. conditions. Given
\begin{equation}
\lambda_{k}^{\prime}\Bigg(\frac{r_{BS,k}^{th}}{B}-\sum_{i=1}^N\log_2\bigg(1+\frac{(1-\rho_k)P_{BS,K}^{i^*}{|h_{BS,K}^i|}^2}{\sigma_{BS,K}^i} \bigg)\Bigg)=0,
\end{equation}
it follows that either $\lambda_{k}^{\prime}=0$ or
\begin{equation}\label{eq:lbd}
\frac{r_{BS,k}^{th}}{B}=\sum_{i=1}^N\log_2\bigg(1+\frac{(1-\rho_k)P_{BS,K}^{i^*}{|h_{BS,K}^i|}^2}{\sigma_{BS,K}^i} \bigg)
\end{equation}
$\lambda_{k}^{\prime}$ cannot be $0$, since otherwise $P_{BS,k}^{i^*}=0$ for all $i\in[1..N]$, which does not meet the rate constraint. Now substituting the expression of the optimal power $P_{BS,k}^{i^*}$, given in Equation~\eqref{eqn:pa}, into Equation~\eqref{eq:lbd} yields
\begin{eqnarray}\label{eqn:lambdaopt}\nonumber
|\mathcal{S}_k|\log_2(\lambda_k^{\prime})=&\frac{r_{BS,k}^{th}}{B}\\
-& \log_2\bigg(\displaystyle{\prod_{i\in \mathcal{S}_k}}\frac{(1-\rho_k){|h_{BS,k}^i|}^2}{(\tilde{\alpha}_k^i-\psi_{k}{|h_{BS,k}^i|}^2)\sigma_{BS,k}^i}\bigg)
\end{eqnarray}
The second K.K.T. condition gives
\begin{equation}
\psi_{k}\bigg(P_{k}^{th}-\sum_{i=1}^NP_{BS,k}^{i^*}{|h_{BS,k}^i|}^2\bigg)=0
\end{equation}
If $\psi_{k}\neq 0$, then  $P_{k}^{th}=\sum_{i=1}^NP_{BS,k}^{i^*}{|h_{BS,k}^i|}^2$ must hold. In this case, substituting $P_{BS,k}^{i^*}$ with its expression given in Equation~\eqref{eqn:pa}, results in
 \begin{eqnarray}\label{eqn:lambda1}\nonumber
P_{k}^{th}&=&\sum_{i\in \mathcal{S}_k}P_{BS,k}^{i^*}{|h_{BS,k}^i|}^2\\
&=&\lambda_{k}^{\prime}\sum_{i\in \mathcal{S}_k}\frac{{|h_{BS,k}^i|}^2}{\tilde{\alpha}_k^i-\psi_{k}{|h_{BS,k}^i|}^2}-\sum_{i\in \mathcal{S}_k}\frac{\sigma_{BS,k}^i}{1-\rho_k}
\end{eqnarray}
Now combining Equations~\eqref{eqn:lambdaopt} and~\eqref{eqn:lambda1} yields
\begin{equation}\label{eqn:pth}
P_{k}^{th}=\frac{2^{r_{BS,k}^{th}/{B|\mathcal{S}_k|}}\displaystyle{\sum_{i\in \mathcal{S}_k}\frac{{|h_{BS,k}^i|}^2}{\tilde{\alpha}_k^i-\psi_{k}{|h_{BS,k}^i|}^2}}}{{\Big(\displaystyle{\prod_{i\in \mathcal{S}_k}}\frac{(1-\rho_k){|h_{BS,k}^i|}^2}{(\tilde{\alpha}_k^i-\psi_{k}{|h_{BS,k}^i|}^2)\sigma_{BS,k}^i}\Big)}^{\frac{1}{|\mathcal{S}_k|}}}-\sum_{i\in \mathcal{S}_k}\frac{\sigma_{BS,k}^i}{1-\rho_k}
\end{equation}
The value of $\psi_k$ that satisfies Equation~\eqref{eqn:pth} is the zero of the function $f$. This ends the proof of the theorem. Note that for consistency, recalling the second K.K.T. condition, $\psi_k=0$ remains a special case of the solution.
\end{proof}

\subsection{Proof of Theorem~\ref{theo:2}}
\begin{proof}
Since the optimization problem~\eqref{eqn:reformulation11} is convex, we rely on the Lagrangian multiplier, which can be written as
\begin{eqnarray}\nonumber
\!\!\!\!\!\!\!\!\mathcal{L}\big(\boldsymbol{P},\{\lambda_{k},\psi_{k}\}_{k=1}^K\big)\!\!\!\!&=&\!\!\!\!\varsigma^{\prime}-\sum_{k=1}^K\lambda_{k}(R_{BS,k}-r_{BS,k}^{th})\\
&-&\!\!\!\!\sum_{k=1}^K\psi_{k}\big(\boldsymbol{P}^T\boldsymbol{h}_k- P_{k}^{th}\big)+\boldsymbol{\tilde{\alpha}}^T\boldsymbol{P}
\end{eqnarray}
For simplicity, let $\gamma_k^i=(1-\rho_k){|h_{BS,k}^i|}^2/\sigma_{BS,k}^i$. By taking the derivative of $\mathcal{L}$ over $P_{BS,k}^i$, it follows that
\begin{equation}
\tilde{\alpha}_k^i-\frac{\lambda_{k}\gamma_{k}^i}{(1+P_{BS,k}^i\gamma_{k}^i)\log(2)}-\sum_{l=1}^K\psi_{l}{|h_{BS,k}^{i,l}|}^2=0
\end{equation}
where $\alpha_k^i=\alpha+\sum_{l}^K\beta_l\eta\rho_l{|h_{BS,k}^{i,l}|}^2$. By letting $\lambda_{k}^{\prime}=\lambda_{k}/\log(2)$, the optimal power level allocated at the BS to user $k$ over the subcarrier $i$ can be derived as
\begin{equation}\label{eqn:pa2}
P_{BS,k}^{i}=\frac{\lambda_{k}^{\prime}}{\tilde{\alpha}_k^i-\sum_{l=1}^K\psi_{l}{|h_{BS,k}^{i,l}|}^2}-\frac{1}{\gamma_{k}^i}
\end{equation}
Then, we restrict the power to be positive.
\end{proof}

\subsection{Proof of Lemma~\ref{lemma2}}
\begin{proof}
Using the K.K.T. condition,
\begin{equation}\label{eqn:kktcond1}
   \lambda_{k}^{\prime}\big(R_{BS,k}-r_{BS,k}^{th}\big)=0,
\end{equation}
and replacing the expression of the optimal power level given by Equation~\eqref{eqn:pa2}, the expression of the K.K.T. multiplier $\lambda_{k}^{\prime}$ can be written as
\begin{eqnarray}
\log_2(\lambda_{k}^{\prime})=\frac{r_{BS,k}^{th}}{B|\mathcal{C}_k|} - \frac{\log_2\bigg(\displaystyle{\prod_{i\in \mathcal{C}_k}}\frac{\gamma_{k}^i}{\tilde{\alpha}_k^i-\sum_{l=1}^{K}\psi_{l}{|h_{BS,k}^{i,l}|}^2}\bigg)}{|\mathcal{C}_k|}
\end{eqnarray}
Hence, we get the expression of $\lambda_{k}^{\prime}$ as in Equation~\eqref{eqn:lambda011}. Hence, getting $\lambda_{k}^{\prime}$ requires the knowledge of $\boldsymbol{\psi}$.
Now, to characterize $\boldsymbol{\psi}$, we consider the second K.K.T condition $\psi_{k}\big(\sum_{l=1}^K\sum_{i=1}^NP_{BS,l}^{i^*}{|h_{BS,l}^{i,k}|}^2-P_{k}^{th}\big)=0$, if $\psi_{k}\neq 0$, then we can write
\begin{eqnarray}\label{eqn:lambda111}\nonumber
P_{k}^{th}\!\!\!&=&\!\!\!\!\sum_{l=1}^K\sum_{i\in\mathcal{C}_{l}}P_{BS,l}^{i^*}{|h_{BS,l}^{i,k}|}^2\\\nonumber
&=&\!\!\!\!\lambda_{k}^{\prime}\sum_{l=1}^{K}\sum_{i\in \mathcal{C}_{l}} \frac{{|h_{BS,l}^{i,k}|}^2}{\tilde{\alpha}_l^i-\sum_{m=1}^{K}\psi_{m}{|h_{BS,l}^{i,m}|}^2}-\sum_{l=1}^K\sum_{i\in \mathcal{C}_{l}}\frac{{|h_{BS,l}^{i,k}|}^2}{\gamma_{l}^i}\\
\end{eqnarray}
Now, substituting Equation~\eqref{eqn:lambda011} into Equation~\eqref{eqn:lambda111} gives
\begin{eqnarray}\nonumber
P_{k}^{th}\!\!\!&=&\!\!\!\!\frac{2^{\frac{r_{BS,k}^{th}}{B|\mathcal{C}_k|}}}{{\bigg(\displaystyle{\prod_{i\in \mathcal{C}_k}}\frac{\gamma_{k}^i}{\tilde{\alpha}_k^i-\sum_{l=1}^{K}\psi_{l}{|h_{BS,k}^{i,l}|}^2}\bigg)}^{\frac{1}{|\mathcal{C}_k|} }}\\ \nonumber
&\times&\!\!\!\!\sum_{l=1}^{K}\sum_{i\in \mathcal{C}_{l}} \frac{{|h_{BS,l}^{i,k}|}^2}{\tilde{\alpha}_l^i-\sum_{m=1}^{K}\psi_{m}{|h_{BS,l}^{i,m}|}^2}-\sum_{l=1}^K\sum_{i\in \mathcal{C}_{l}}\frac{{|h_{BS,l}^{i,k}|}^2}{\gamma_{l}^i}.
\end{eqnarray}
Hence, it is clear that $\boldsymbol{\psi}$ is the zero of the functions $f_{k}$
\end{proof}

\bibliographystyle{IEEEtran}
\bibliography{References}

\begin{IEEEbiographynophoto}{Bassem Khalfi (S'14)} %
is currently a Ph.D student at Oregon State University.
His research focuses on various topics in the area of wireless communication and networks, including dynamic spectrum access, energy harvesting, and IoT.
\end{IEEEbiographynophoto}
\begin{IEEEbiographynophoto}
{Bechir Hamdaoui (S'02-M'05-SM'12)}
is an Associate Professor in the School of EECS at Oregon State University.
His research interest spans various areas in the fields of computer networking, wireless communications, and mobile computing.
\end{IEEEbiographynophoto}
\begin{IEEEbiographynophoto}
{Mahdi Ben Ghorbel (S'10-M'14)}
is currently a postdoctoral researcher at The University of British Columbia, Okanagan Campus in BC, Canada since March 2016.
His research interests include optimization of resource allocation for next generation communication systems.
\end{IEEEbiographynophoto}
\begin{IEEEbiographynophoto}
{Mohsen Guizani (S'85-M'89-SM'99-F'09)}
is currently a Professor in the ECE Department  at the University of Idaho, USA.
His research interests include wireless communications and mobile computing, computer networks, mobile cloud computing, security, and smart grid.
\end{IEEEbiographynophoto}
\begin{IEEEbiographynophoto}
{Xi Zhang (S'89-SM'98-F'16)}
is a Professor at Texas A\&M University. His research interest includes
cognitive and cooperative radio networks, underwater wireless communications and networks, and mobile cloud
computing.
\end{IEEEbiographynophoto}
\begin{IEEEbiographynophoto}
{Nizar Zorba (M'08)}
is an Associate Professor in the EE department at Qatar University, Doha, Qatar.
His research interests span 5G networks optimization, demand-response in smart grids, and crowd management.
\end{IEEEbiographynophoto}

\end{document}